\documentclass[fullpaper,9pt]{IEEEtran}

\usepackage{amsmath,amssymb,enumerate,subfigure,multirow,hhline,color}
\usepackage{xcolor}
\usepackage{hyperref}
\usepackage{graphicx}
\usepackage{graphics}
\usepackage[sort]{cite}
\usepackage{amsthm}
\usepackage{algorithm}
\usepackage{algpseudocode}

\usepackage{tcolorbox}
\usepackage{epsfig,psfrag}
\usepackage{tabularx}
\usepackage{tabulary}

\usepackage[sort]{cite}

\usepackage{hyperref}
\hypersetup{breaklinks=true,colorlinks=true,linkcolor=blue,citecolor=blue}

\usepackage[noabbrev,capitalise,nameinlink]{cleveref}

\DeclareMathOperator{\argmax}{arg\,max}

\newtheorem{theorem}{Theorem}

\newtheorem{assumption}{Assumption}
\newtheorem{remark}{Remark}
\newtheorem{corollary}{Corollary}

\newtheorem{lemma}{Lemma}
\newtheorem{definition}{Definition}
\newtheorem{proposition}{Proposition}

\allowdisplaybreaks

\title{Final Iteration Convergence Bound of Q-Learning: Switching System Approach}

\author{Donghwan Lee
\thanks{D. Lee is with the Department of Electrical Engineering,
KAIST, Daejeon, 34141, South Korea {\tt\small
donghwan@kaist.ac.kr}.}
}

\begin{document}

\maketitle

\begin{abstract}

Q-learning is known as one of the fundamental reinforcement learning (RL) algorithms. Its convergence has been the focus of extensive research over the past several decades. Recently, a new finite-time error bound and analysis for Q-learning was introduced using a switching system framework. This approach views the dynamics of Q-learning as a discrete-time stochastic switching system. The prior study established a finite-time error bound on the averaged iterates using Lyapunov functions, offering further insights into Q-learning. While valuable, the analysis focuses on error bounds of the averaged iterate, which comes with the inherent disadvantages: it necessitates extra averaging steps, which can decelerate the convergence rate. Moreover, the final iterate, being the original format of Q-learning, is more commonly used and is often regarded as a more intuitive and natural form in the majority of iterative algorithms. In this paper, we present a finite-time error bound on the final iterate of Q-learning based on the switching system framework. The proposed error bounds have different features compared to the previous works, and cover different scenarios. Finally, we expect that the proposed results provide additional insights on Q-learning via connections with discrete-time switching systems, and can potentially present a new template for finite-time analysis of more general RL algorithms.
\end{abstract}
\begin{IEEEkeywords}
Reinforcement learning, Q-learning, switching system, convergence, finite-time analysis
\end{IEEEkeywords}

\section{Introduction}

Reinforcement learning (RL) addresses the optimal sequential decision making problem for unknown systems through
experiences~\cite{sutton1998reinforcement}. Recent successes of RL algorithms outperforming humans in several challenging tasks~\cite{mnih2015human,wang2016dueling,lillicrap2016continuous,heess2015memory,van2016deep,bellemare2017distributional,schulman2015trust,schulman2017proximal}
have triggered a surge of interests in RL both theoretically and
experimentally. Among many others, Q-learning~\cite{watkins1992q} is one of the most fundamental and popular RL algorithms, and its convergence has been extensively studied over the past decades. Classical analysis mostly focuses on asymptotic convergence~\cite{tsitsiklis1994asynchronous,jaakkola1994convergence,borkar2000ode,hasselt2010double,melo2008analysis,lee2020unified,devraj2017zap}.
While crucial, the asymptotic convergence cannot measure the speed at which iterations approach a solution. Consequently, the efficiency of the related algorithms cannot be precisely assessed. For this reason, finite-time convergence analysis, which quantifies how fast the iterations progress toward the solution, has gained increasing attention recently.

Recently, advances have been made in finite-time convergence analysis~\cite{szepesvari1998asymptotic,kearns1999finite,even2003learning, azar2011speedy,beck2012error,wainwright2019stochastic,qu2020finite,li2020sample,chen2021lyapunov,lee2020periodic,lee2021discrete}. Most of the existing results treat the Q-learning dynamics as nonlinear stochastic approximations~\cite{kushner2003stochastic}, and use the contraction property of the Bellman equation. Recently,~\cite{lee2021discrete} proposed a new perspective of Q-learning based on discrete-time switching system models~\cite{liberzon2003switching,lin2009stability}, and established a finite-time analysis based on tools in control theory~\cite{chen1995linear,khalil2002nonlinear}. The switching system perspective captures unique features of Q-learning dynamics, and allows us to convert the notion of finite-time convergence analysis into the stability analysis of dynamic control systems. While valuable, the analysis in~\cite{lee2021discrete} focuses on error bounds of the {\em average iterate}, which comes with the inherent disadvantages: it necessitates extra averaging steps, which can decelerate the convergence rate. Moreover, the final iterate, being the original format of Q-learning, is more commonly used and is often regarded as a more intuitive and natural form in the majority of iterative algorithms. Therefore, it is more interesting to study the convergence and finite-time analysis of the final iterate.

Given these considerations, the main goal of this paper is to present a finite-time error bound on the {\em final iterate} of Q-learning based on the switching system framework in~\cite{lee2021discrete} for additional insights and complementary analysis.
In particular, we improve the analysis in~\cite{lee2021discrete} by replacing the average iterate with the final iterate and deriving the following bound:
\begin{align}
{\mathbb E}[\left\| {Q_k  - Q^* } \right\|_\infty  ] \le& \frac{{9 d_{\max } |{\cal S} \times {\cal A}|\alpha ^{1/2} }}{{d_{\min }^{3/2} (1 - \gamma )^{5/2} }} + \frac{{2|{\cal S} \times {\cal A}|^{3/2} }}{{1 - \gamma }}\rho ^k \nonumber \\
& + \frac{{4\gamma d_{\max } |{\cal S} \times {\cal A}|^{2/3} }}{{1 - \gamma }}\frac{1}{{d_{\min } (1 - \gamma )}}\rho ^{k/2 - 1},\label{eq:proposed-error-bound}
\end{align}
where $|{\cal S} \times {\cal A}|$ is the number of the state-action pairs, $\gamma$ is the discount factor, $d_{\min}$ is the minimum state-action occupation frequency, $\alpha \in (0,1)$ is the constant step-size, $Q^*$ is a vectorized optimal Q-function, $Q_k$ is its estimation at the current time $k$, and $\rho := 1 - \alpha d_{\min} (1-\gamma) \in (0,1)$ is the exponential decay rate. Moreover, the sample complexity, $\tilde {\cal O}\left( \frac{d_{\max}^4 |{\cal S}\times {\cal A}|^4}{\varepsilon^4 d_{\min}^6 (1-\gamma)^{10}} \right)$, in~\cite{lee2021discrete} for $\varepsilon$-optimal solution can be improved to $\tilde {\cal O}\left( {\frac{{\gamma ^2 d_{\max }^2 |{\cal S} \times {\cal A}|^2 }}{{\varepsilon ^2 d_{\min }^4 (1 - \gamma )^6 }}} \right)$ from the proposed approach, where $d_{\max}$ is the maximum state-action occupation frequency, and $\tilde {\cal O}$ ignores the constant and polylogarithmic factors. We note that this extension is not trivial, and significantly different approaches have been adopted in this paper.

The proposed analysis relies on propagations of the autocorrelation matrix instead of the Lyapunov function analysis used in~\cite{lee2021discrete}, which allows conceptually simpler analysis. It also provides additional insights on Q-learning via connections with discrete-time switching systems, and can potentially present a new template for finite-time analysis of more general RL algorithms. Moreover, the new perspective can potentially stimulate synergy between control theory and RL, and open up opportunities to the design of new RL algorithms. The proposed error bounds have different features compared to the previous works, and cover different cases detailed throughout this paper. Finally, we note that this paper only covers an i.i.d. observation model with constant step-sizes for simplicity of the overall analysis. Extensions to more complicated scenarios are not the main purpose of this paper.

\paragraph*{Related Works}
Recently, some progresses have been made in finite-time analysis of Q-learning~\cite{szepesvari1998asymptotic,kearns1999finite,even2003learning,azar2011speedy,beck2012error,wainwright2019stochastic,qu2020finite,li2020sample,lee2020periodic,chen2021lyapunov,lee2021discrete}. In particular,~\cite{szepesvari1998asymptotic} provided a finite-time convergence rate with state-action dependent diminishing step-sizes. The authors in~\cite{kearns1999finite} analyzed a batch version of synchronous Q-learning, called phased Q-learning, with finite-time bounds. The authors of~\cite{even2003learning} developed convergence rates for both synchronous and asynchronous Q-learning with polynomial and linear step-sizes. \cite{azar2011speedy} proposed a variant of synchronous Q-learning called speedy Q-learning by adding a momentum term, and obtained an accelerated learning rate. A finite-time analysis of asynchronous Q-learning with constant step-sizes was considered in~\cite{beck2012error}. Afterwards, many advances have been made recently in finite-time analysis. The paper~\cite{lee2020periodic} developed the so-called periodic Q-learning mimicking the stochastic gradient-based training scheme in~\cite{mnih2015human} with periodic target updates. The paper~\cite{wainwright2019stochastic} provided finite-time bounds for general synchronous stochastic approximation, and applied it to a synchronous Q-learning with state-independent diminishing step-sizes. In~\cite{qu2020finite}, a finite-time convergence rate of general asynchronous stochastic approximation scheme was derived, and it was applied to asynchronous Q-learning with diminishing step-sizes. Subsequently,~\cite{li2020sample} obtained sharper bounds under constant step-sizes,~\cite{chen2021lyapunov} provided a Lyapunov method-based analysis for general stochastic approximations and Q-learning with both constant and diminishing step-sizes, and~\cite{lee2021discrete} proposed a switching system perspective of Q-learning, and established a finite-time analysis.

\section{Preliminaries}\label{sec:preliminaries}

\subsection{Notation}
The adopted notation is as follows: ${\mathbb R}$: set of real numbers; ${\mathbb R}^n $: $n$-dimensional Euclidean
space; ${\mathbb R}^{n \times m}$: set of all $n \times m$ real
matrices; $A^T$: transpose of matrix $A$; $A \succ 0$ ($A \prec
0$, $A\succeq 0$, and $A\preceq 0$, respectively): symmetric
positive definite (negative definite, positive semi-definite, and
negative semi-definite, respectively) matrix $A$; $I$: identity matrix with appropriate dimensions; $\lambda_{\min}(A)$ and $\lambda_{\max}(A)$ for any symmetric matrix $A$: the minimum and maximum eigenvalues of $A$; $|{\cal S}|$: cardinality of a finite set $\cal S$; ${\rm tr}(A)$: trace of any matrix $A$; $A \otimes B$: Kronecker product of matrices $A$ and $B$.

\subsection{Markov decision problem}
We consider the infinite-horizon discounted Markov decision problem (MDP) and Markov decision process, where the agent sequentially takes actions to maximize cumulative discounted rewards. In a Markov decision process with the state-space ${\cal S}:=\{ 1,2,\ldots ,|{\cal S}|\}$ and action-space ${\cal A}:= \{1,2,\ldots,|{\cal A}|\}$, the decision maker selects an action $a \in {\cal A}$ at the current state $s$, then the state
transits to the next state $s'$ with probability $P(s,a,s')$, and the transition incurs a
reward $r(s,a,s')$, where $P(s,a,s')$ is the state transition probability from the current state
$s\in {\cal S}$ to the next state $s' \in {\cal S}$ under action $a \in {\cal A}$, and $r(s,a,s')$ is the reward function. For convenience, we consider a deterministic reward function and simply write $r(s_k,a_k ,s_{k + 1}) =:r_k,k \in \{ 0,1,\ldots \}$.

A deterministic policy, $\pi :{\cal S} \to {\cal A}$, maps a state $s \in {\cal S}$ to an action $\pi(s)\in {\cal A}$. The objective of the Markov decision problem (MDP) is to find a deterministic optimal policy, $\pi^*$, such that the cumulative discounted rewards over infinite time horizons is
maximized, i.e.,
\begin{align*}
\pi^*:= \argmax_{\pi\in \Theta} {\mathbb E}\left[\left.\sum_{k=0}^\infty {\gamma^k r_k}\right|\pi\right],
\end{align*}
where $\gamma \in [0,1)$ is the discount factor, $\Theta$ is the set of all deterministic policies, $(s_0,a_0,s_1,a_1,\ldots)$ is a state-action trajectory generated by the Markov chain under policy $\pi$, and ${\mathbb E}[\cdot|\pi]$ is an expectation conditioned on the policy $\pi$. Q-function under policy $\pi$ is defined as
\begin{align*}
Q^{\pi}(s,a)={\mathbb E}\left[ \left. \sum_{k=0}^\infty {\gamma^k r_k} \right|s_0=s,a_0=a,\pi \right], (s,a)\in {\cal S} \times {\cal A},
\end{align*}
and the optimal Q-function is defined as $Q^*(s,a)=Q^{\pi^*}(s,a)$ for all $(s,a)\in {\cal S} \times {\cal A}$. Once $Q^*$ is known, then an optimal policy can be retrieved by the greedy policy $\pi^*(s)=\argmax_{a\in {\cal A}}Q^*(s,a)$. Throughout, we assume that the MDP is ergodic so that the stationary state distribution exists and the Markov decision problem is well posed.

\subsection{Switching system}

Since a switching system~\cite{liberzon2003switching,lin2009stability} is a special form of nonlinear systems~\cite{khalil2002nonlinear}, we first consider the nonlinear system
\begin{align}
x_{k+1}=f(x_k),\quad x_0=z \in {\mathbb R}^n,\quad k\in \{0,1,\ldots \},\label{eq:nonlinear-system}
\end{align}
where $x_k\in {\mathbb R}^n$ is the state and $f:{\mathbb R}^n \to {\mathbb R}^n$ is a nonlinear mapping. An important concept in dealing with the nonlinear system is the equilibrium point. A point $x=x^*$ in the state-space is said to be an equilibrium point of~\eqref{eq:nonlinear-system} if it has the property that whenever the state of the system starts at $x^*$, it will remain at $x^*$~\cite{khalil2002nonlinear}. For~\eqref{eq:nonlinear-system}, the equilibrium points are the real solutions of the equation $f(x) = x$. The equilibrium point $x^*$ is said to be globally asymptotically stable if for any initial state $x_0 \in {\mathbb R}^n$, $x_k \to x^*$ as $k \to \infty$.

Next, let us consider the particular system, called the \emph{linear switching system},
\begin{align}
&x_{k+1}=A_{\sigma_k} x_k,\quad x_0=z\in {\mathbb
R}^n,\quad k\in \{0,1,\ldots \},\label{eq:switched-system}
\end{align}
where $x_k \in {\mathbb R}^n$ is the state, $\sigma\in {\mathcal M}:=\{1,2,\ldots,M\}$ is called the mode, $\sigma_k \in
{\mathcal M}$ is called the switching signal, and $\{A_\sigma,\sigma\in {\mathcal M}\}$ are called the subsystem matrices. The switching signal can be either arbitrary or controlled by the user under a certain switching policy. Especially, a state-feedback switching policy is denoted by $\sigma_k = \sigma(x_k)$. A more general class of systems is the {\em affine switching system}
\begin{align}
&x_{k+1}=A_{\sigma_k} x_k + b_{\sigma_k},\quad x_0=z\in {\mathbb
R}^n,\quad k\in \{0,1,\ldots \},\label{eq:affine-switching-system}
\end{align}
where $b_{\sigma_k} \in {\mathbb R}^n$ is the additional input vector, which also switches according to $\sigma_k$. Due to the additional input $b_{\sigma_k}$, its stabilization becomes much more challenging.

\subsection{Assumptions and definitions}
\begin{algorithm}[t]
\caption{Q-Learning with a constant step-size}
  \begin{algorithmic}[1]
    \State Initialize $Q_0 \in {\mathbb R}^{|{\cal S}||{\cal A}|}$ randomly such that $\left\| {Q_0 } \right\|_\infty   \le 1$.
    \State Sample $s_0\sim p$
    \For{iteration $k=0,1,\ldots$}
    	\State Sample $a_k\sim \beta(\cdot|s_k)$ and $s_k\sim p(\cdot)$
        \State Sample $s_k'\sim P(s_k,a_k,\cdot)$ and $r_k= r(s_k,a_k,s_k')$
        \State Update $Q_{k+1}(s_k,a_k)=Q_k(s_k,a_k)+\alpha \{r_k+\gamma\max_{u \in {\cal A}} Q_k(s_k',u)-Q_k (s_k,a_k)\}$
    \EndFor
  \end{algorithmic}\label{algo:standard-Q-learning2}
\end{algorithm}
In this paper, we focus on the standard Q-learning in~\cref{algo:standard-Q-learning2} with a constant step-size $\alpha \in (0,1)$ under the following setting: $\{(s_k,a_k,s_k')\}_{k=0}^{\infty}$ are i.i.d. samples under the behavior policy $\beta$, where the time-invariant behavior policy is the policy by which the RL agent actually behaves to collect experiences. Note that the notation $s_k'$ implies the next state sampled at the time step $k$, which is used instead of $s_{k+1}$ in order to distinguish $s_k'$ from $s_{k+1}$. In this paper, the notation $s_{k+1}$ indicate the current state at the iteration step $k+1$, while it does not depend on $s_k$. For simplicity, we assume that the state at each time is sampled from the stationary state distribution $p$, and in this case, the state-action distribution at each time is identically given by
\begin{align*}
d(s,a) = p (s)\beta (a|s),\quad (s,a) \in {\cal S} \times {\cal A}.
\end{align*}
\begin{remark}
In this paper, we assume that the behavior policy $\beta$ is time-invariant, and this scenario excludes the common method of using the $\varepsilon$-greedy behavior policy with $\varepsilon > 0$ because the $\varepsilon$-greedy behavior policy depends on the current Q-iterate, and hence is time-varying. Moreover, the proposed analysis cannot be easily extended to the analysis of Q-learning with the $\varepsilon$-greedy behavior policy due to reasons that will appear later in this paper.
\end{remark}

Throughout, we make the following assumptions for convenience.
\begin{assumption}\label{assumption:positive-distribution}
$d(s,a)> 0$ holds for all $(s,a)\in {\cal S} \times {\cal A}$.
\end{assumption}
\begin{assumption}\label{assumption:step-size}
The step-size is a constant $\alpha \in (0,1)$.
\end{assumption}
\begin{assumption}\label{assumption:bounded-reward} The reward is bounded as follows:
\begin{align*}
\max _{(s,a,s') \in {\cal S} \times {\cal A} \times {\cal S}} |r (s,a,s')| =:R_{\max}\leq 1.
\end{align*}
\end{assumption}
\begin{assumption}\label{assumption:bounded-Q0} The initial iterate $Q_0$ satisfies $\left\| {Q_0 } \right\|_\infty   \le 1$.
\end{assumption}
\begin{remark}
All these assumptions will be used throughout this paper for convergence proof. \cref{assumption:positive-distribution} guarantees that every state-action pair is visited infinitely often for sufficient exploration. This assumption is used when the state-action occupation frequency is given. It has been also considered in~\cite{li2020sample} and~\cite{chen2021lyapunov}. The work in~\cite{beck2012error} considers another exploration condition, called the cover time condition, which states that there is a certain time period, within which every state-action pair is expected to be visited at least once. Slightly different cover time conditions have been used in~\cite{even2003learning} and~\cite{li2020sample} for convergence rate analysis. The unit bounds imposed on $R_{\max}$ and $Q_0$ are just for simplicity of analysis. The constant step-size in~\cref{assumption:step-size} has been also studied in~\cite{beck2012error} and~\cite{chen2021lyapunov} using different approaches.
\end{remark}

The following quantities will be frequently used in this paper; hence, we define the corresponding notations for convenience.
\begin{definition}
\begin{enumerate}
\item Maximum state-action occupation frequency:
\[
d_{\max} := \max_{(s,a)\in {\cal S} \times {\cal A}} d(s,a) \in (0,1).
\]

\item Minimum state-action occupation frequency:
\[
d_{\min}:= \min_{(s,a) \in {\cal S} \times {\cal A}} d(s,a) \in (0,1).
\]

\item Exponential decay rate:
\begin{align}\label{eq:rho}
    \rho:=1 - \alpha d_{\min} (1-\gamma).
\end{align}
Under~\cref{assumption:step-size}, the decay rate satisfies $\rho \in (0,1)$.
\end{enumerate}
\end{definition}

Throughout the paper, we will use the following matrix notations for compact dynamical system representations:
\begin{align*}
P:=& \begin{bmatrix}
   P_1\\
   \vdots\\
   P_{|{\cal A}|}\\
\end{bmatrix},\; R:= \begin{bmatrix}
   R_1 \\
   \vdots \\
   R_{|{\cal A}|} \\
\end{bmatrix},
\; Q:= \begin{bmatrix}
   Q(\cdot,1)\\
  \vdots\\
   Q(\cdot,|{\cal A}|)\\
\end{bmatrix},\\
D_a:=& \begin{bmatrix}
   d(1,a) & & \\
   & \ddots & \\
   & & d(|{\cal S}|,a)\\
\end{bmatrix},
\; D:= \begin{bmatrix}
   D_1 & & \\
    & \ddots  & \\
    & & D_{|{\cal A}|} \\
\end{bmatrix},
\end{align*}
where $P_a = P(\cdot,a,\cdot)\in {\mathbb R}^{|{\cal S}| \times |{\cal S}|}$, $Q(\cdot,a)\in {\mathbb R}^{|{\cal S}|},a\in {\cal A}$ and $R_a(s):={\mathbb E}[r(s,a,s')|s,a]$.
Note that $P\in{\mathbb R}^{|{\cal S}\times {\cal A}| \times |{\cal S}|  }$, $R \in {\mathbb R}^{|{\cal S}\times {\cal A}|}$, $Q\in {\mathbb R}^{|{\cal S}\times {\cal A}|}$, and $D\in {\mathbb R}^{|{\cal S}\times {\cal A}| \times |{\cal S}\times {\cal A}|}$.
In this notation, Q-function is encoded as a single vector $Q \in {\mathbb R}^{|{\cal S}\times {\cal A}|}$, which enumerates $Q(s,a)$ for all $s \in {\cal S}$ and $a \in {\cal A}$. In particular, the single value $Q(s,a)$ can be written as
\begin{align*}
Q(s,a) = (e_a  \otimes e_s )^T Q,
\end{align*}
where $e_s \in {\mathbb R}^{|{\cal S}|}$ and $e_a \in {\mathbb R}^{|{\cal A}|}$ are $s$-th basis vector (all components are $0$ except for the $s$-th component which is $1$) and $a$-th basis vector, respectively. Note also that under~\cref{assumption:positive-distribution}, $D$ is a nonsingular diagonal matrix with strictly positive diagonal elements.

For any stochastic policy, $\pi:{\cal S}\to \Delta_{|{\cal A}|}$, where $\Delta_{|{\cal A}|}$ is the set of all probability distributions over ${\cal A}$, we define the corresponding action transition matrix as
\begin{align}
\Pi^\pi:=\begin{bmatrix}
   \pi(1)^T \otimes e_1^T\\
   \pi(2)^T \otimes e_2^T\\
    \vdots\\
   \pi(|S|)^T \otimes e_{|{\cal S}|}^T \\
\end{bmatrix}\in {\mathbb R}^{|{\cal S}| \times |{\cal S}\times {\cal A}|},\label{eq:swtiching-matrix}
\end{align}
where $e_s \in {\mathbb R}^{|{\cal S}|}$.
Then, it is well known that
$
P\Pi^\pi \in {\mathbb R}^{|{\cal S}\times {\cal A}| \times |{\cal S}\times {\cal A}|}
$
is the transition probability matrix of the state-action pair under policy $\pi$.
If we consider a deterministic policy, $\pi:{\cal S}\to {\cal A}$, the stochastic policy can be replaced with the corresponding one-hot encoding vector
$
\vec{\pi}(s):=e_{\pi(s)}\in \Delta_{|{\cal A}|},
$
where $e_a \in {\mathbb R}^{|{\cal A}|}$, and the corresponding action transition matrix is identical to~\eqref{eq:swtiching-matrix} with $\pi$ replaced with $\vec{\pi}$. For any given $Q \in {\mathbb R}^{|{\cal S}\times {\cal A}|}$, denote the greedy policy w.r.t. $Q$ as $\pi_Q(s):=\argmax_{a\in {\cal A}} Q(s,a)\in {\cal A}$.
We will use the following shorthand frequently:
$\Pi_Q:=\Pi^{\pi_Q}.$

The boundedness of Q-learning iterates~\cite{gosavi2006boundedness} plays an important role in our analysis.
\begin{lemma}[Boundedness of Q-learning iterates~\cite{gosavi2006boundedness}]\label{lemma:bounded-Q}
If the step-size is less than one, then for all $k \ge 0$,
\begin{align*}
\|Q_k\|_\infty \le Q_{\max}:= \frac{\max \{R_{\max},\max_{(s,a)\in {\cal S} \times {\cal A}} Q_0 (s,a)\}}{1-\gamma}.
\end{align*}
From~\cref{assumption:bounded-reward} and~\cref{assumption:bounded-Q0}, we can easily see that $Q_{\max}\leq\frac{1}{1-\gamma}$.
\end{lemma}

\section{Finite-time Analysis of Q-learning from Switching System Theory}\label{sec:convergence}

In this section, we study a discrete-time switching system model of Q-learning in~\cref{algo:standard-Q-learning2}, and establish its finite-time convergence bound based on the stability analysis of switching system.

\subsection{Q-learning as a stochastic affine switching system}
Using the notation introduced, the update in~\cref{algo:standard-Q-learning2} can be rewritten as
\begin{align}
Q_{k+1}=Q_k+\alpha \{DR+\gamma DP\Pi_{Q_k}Q_k-DQ_k +w_{k}\},\label{eq:1}
\end{align}
where
\begin{align}
w_{k}=&(e_{a_k}\otimes e_{s_k} ) r_k+\gamma (e_{a_k}\otimes e_{s_k} )(e_{s_k'})^T \Pi_{Q_k}Q_k\nonumber\\
&-(e_{a_k}\otimes e_{s_k})(e_{a_k}\otimes e_{s_k})^T Q_k\nonumber\\
& -(DR+\gamma DP\Pi_{Q_k}Q_k-DQ_k)\nonumber\\
=& (e_a  \otimes e_s )\delta _k  - (DR + \gamma DP\Pi _{Q_k } Q_k  - DQ_k ),\label{eq:w}
\end{align}
and
\begin{align}
\delta _k : = r_k  + \gamma (e_{s'} )^T \Pi _{Q_k } Q_k  - (e_a  \otimes e_s )^T Q_k\label{eq:td-error}
\end{align}
is the so-called temporal-difference (TD) error~\cite{sutton1988learning}, and $(s_k,a_k,r_k,s_k')$ is the sample in the $k$-th time-step. Note that by definition, the noise term has a zero mean conditioned on $Q_k$, i.e., ${\mathbb E}[w_k|Q_k]=0$.
Recall the definitions $\pi_Q(s)$ and $\Pi_Q$. Invoking the optimal Bellman equation $(\gamma DP\Pi_{Q^*}-D)Q^*+DR=0$,~\eqref{eq:1} can be further rewritten by
\begin{align}
(Q_{k+1} - Q^*) =& \{ I + \alpha (\gamma DP\Pi _{Q_k} - D)\}(Q_k - Q^*)\nonumber\\
&+\gamma DP(\Pi_{Q_k} - \Pi_{Q^*})Q^* + \alpha w_k.\label{eq:Q-learning-stochastic-recursion-form}
\end{align}
which is a linear switching system with an extra affine term, $\gamma DP(\Pi_{Q_k} - \Pi_{Q^*})Q^*$, and a stochastic noise vector, $w_k$. For any $Q \in {\mathbb R}^{|{\cal S}\times {\cal A}|}$, define
\begin{align*}
A_Q :=I + \alpha(\gamma DP\Pi_Q-D),\quad b_Q:=\gamma DP(\Pi_{Q} - \Pi_{Q^*})Q^*.
\end{align*}
Using the notation, the Q-learning iteration can be concisely represented as the \emph{stochastic affine switching system}
\begin{align}
Q_{k + 1}- Q^* = A_{Q_k} (Q_k - Q^*) + b_{Q_k} + \alpha w_k,\label{eq:swithcing-system-form}
\end{align}
where $A_{Q_k} \in {\mathbb R}^{|{\cal S}\times {\cal A}| \times |{\cal S}\times {\cal A}|}$ and $b_{Q_k}\in {\mathbb R}^{|{\cal S}\times {\cal A}|}$ switch among matrices from $\{I + \alpha (\gamma DP\Pi^\pi-D):\pi\in \Theta\}$ and vectors from $\{\gamma DP(\Pi^\pi - \Pi^{\pi^*})Q^* :\pi\in\Theta \}$. In particular, let us define a one-to-one mapping $\varphi :\Theta  \to \{ 1,2, \ldots ,|\Theta |\}$ from a deterministic policy $\pi \in \Theta$ to an integer in $\{ 1,2, \ldots ,|\Theta |\}$, and define
\begin{align}
{A_i} =& I + \alpha (\gamma DP{\Pi ^\pi } - D) \in {\mathbb R}^{|{\cal S}\times {\cal A}| \times |{\cal S} \times {\cal A}|},\label{eq:12}\\
{b_i} =& \gamma DP(\Pi^\pi - \Pi^{\pi^*})Q^* \in {\mathbb R}^{|{\cal S}\times {\cal A}|},\nonumber
\end{align}
for all $i = \varphi (\pi )$ and $\pi  \in \Theta$. Then,~\eqref{eq:swithcing-system-form} can be written by the affine switching system~\eqref{eq:affine-switching-system} with the switching signal ${\sigma _k} \in \{ 1,2, \ldots ,|\Theta |\}$ at time $k\geq 0$  determined by ${\sigma _k} = \varphi ({\pi _k})$ with ${\pi _k}(\cdot): = \arg {\min _{a \in A}}{Q_k}( \cdot ,a) \in \Theta$.

Consequently, the convergence of Q-learning is now reduced to analyzing the stability of the above switching system.
A main obstacle in proving the stability  arises from the presence of the affine and stochastic terms. Without these terms, we can easily establish the exponential stability of the corresponding deterministic switching system, under arbitrary switching policy. Specifically, we have the following result.
\begin{proposition}[\cite{lee2021discrete}]\label{prop:stability}
For arbitrary $H_k \in {\mathbb R}^{|{\cal S}\times {\cal A}|}, k\ge 0$, the linear switching system
\begin{align*}
Q_{k+1} - Q^* &= A_{H_k} (Q_k - Q^*),\quad  Q_0 - Q^*\in {\mathbb R}^{|{\cal S}||{\cal A}|},
\end{align*}
is exponentially stable with $\|Q_k- Q^*\|_\infty\le \rho ^k \|Q_0 - Q^*\|_\infty, k \ge 0$, where $\rho$ is defined in~\eqref{eq:rho}.
\end{proposition}

The above result follows immediately from the key fact that $\|A_Q \|_\infty \le \rho$, which we formally state in the lemma below.

\begin{lemma}[\cite{lee2021discrete}]\label{lemma:max-norm-system-matrix}
For any $Q \in {\mathbb R}^{|{\cal S}\times {\cal A}|}$, $\|A_Q \|_\infty \le \rho$, where $\| A \|_\infty :=\max_{1\le i \le m} \sum_{j=1}^n {|A_{ij} |}$ and $A_{ij}$ is the element of $A$ in $i$-th row and $j$-th column.
\end{lemma}

However, because of the additional affine term and stochastic noises in the original switching system~\eqref{eq:swithcing-system-form}, it is not obvious how to directly derive its finite-time convergence bound. To circumvent the difficulty with the affine term, we will resort to two simpler comparison systems, whose trajectories upper and lower bound that of the original system, and can be more easily analyzed. These systems will be called the upper and lower comparison systems, which capture important behaviors of Q-learning. The upper comparison system, denoted by $Q_k^U$, upper bounds Q-learning iterate $Q_k$, while the lower comparison system, denoted by $Q_k^L$, lower bounds $Q_k$. The construction of these comparison systems is partly inspired by~\cite{lee2020unified} and exploits the special structure of the Q-learning algorithm. Unlike~\cite{lee2020unified}, here we focus on the discrete-time domain and a finite-time analysis. To address the difficulty with the stochastic noise, we introduce a two-phase analysis: the first phase captures the noise effect of the lower comparison system, while the second phase captures the difference between the two comparison systems when noise effect vanishes.

\begin{remark}
When an $\varepsilon$-greedy strategy is utilized for the behavior policy, it results in the behavior policy that is time-varying and depends on the current Q-iterate, $Q_k$. This implies that the matrix $D$ in~\eqref{eq:Q-learning-stochastic-recursion-form} becomes a time-varying matrix depending on $Q_k$, introducing further nonlinearity and probabilistic dependencies within the switching system dynamics in~\eqref{eq:Q-learning-stochastic-recursion-form}. Consequently, it is not feasible to straightforwardly extend the proposed analysis to Q-learning with the $\varepsilon$-greedy behavior policy. A more extensive analysis is necessary for this extension, which is left as a subject for future research endeavors.
\end{remark}

Before closing this section, we present the following result which will be useful throughout the paper.
\begin{lemma}\label{lemma:nonnegative-matrix}
For any $Q \in {\mathbb R}^{|{\cal S}\times {\cal A}|}$, $A_Q$ is a nonnegative matrix (all entries are nonnegative).
\end{lemma}
\begin{proof}
Recalling the definition $A_Q :=I + \alpha(\gamma DP\Pi_Q-D)$, one can easily see that for any $i,j \in \{ 1,2, \ldots ,|{\cal S}\times {\cal A}|\} $, we have ${[{A_Q}]_{ij}} = {[I - \alpha D + \alpha \gamma DP{\Pi _Q}]_{ij}} = {[I - \alpha D]_{ij}} + \alpha \gamma {[DP{\Pi _Q}]_{ij}} \ge 0$, where ${[ \cdot ]_{ij}}$ denotes the element of a matrix $[ \cdot ]$ in the $i$th row and $j$th column, and the inequality follows from the fact that both $I - \alpha D$ and $DP{\Pi _Q}$ are nonnegative matrices. This completes the proof.
\end{proof}

\subsection{Lower comparison system}
Let us consider the stochastic linear system~\cite{lee2021discrete}
\begin{align}
Q_{k+1}^L- Q^*=A_{Q^*} (Q_k^L - Q^*) + \alpha w_{k},\quad  Q_{0}^L-Q^* \in {\mathbb R}^{|{\cal S}\times {\cal A}|},\label{eq:lower-system}
\end{align}
 where the stochastic noise $w_{k}$ is the same as the original system~\eqref{eq:Q-learning-stochastic-recursion-form}. We call it the \emph{lower comparison system}.
\begin{proposition}[{\cite{lee2021discrete}}]\label{prop:lower-bound2}
Suppose $Q_0^L- Q^*\le Q_0 - Q^*$, where $\le$ is used as the element-wise inequality. Then, $Q_k^L- Q^*\le Q_k-Q^*$ for all $k\geq0$.
\end{proposition}
\begin{proof}
The proof is done by an induction argument. Suppose the result holds for some $k \ge 0$. Then, we have
\begin{align*}
(Q_{k+1}- Q^*)=&A_{Q^*} (Q_k-Q^*)+(A_{Q_k}-A_{Q^*}) (Q_k-Q^*)\\
&+b_{Q_k}+ \alpha w_k\\
=& A_{Q^*} (Q_k-Q^*) + \alpha\gamma DP(\Pi_{Q_k}-\Pi_{Q^*})Q_k \\
&+\alpha w_k\ge A_{Q^*} (Q_k-Q^*) + \alpha w_k\\
\ge& A_{Q^*} (Q_k^L-Q^*) + \alpha w_k\\
=& Q_{k+1}^L-Q^*,
\end{align*}
where the first inequality is due to $DP(\Pi_{Q_k}- \Pi_{Q^*})Q_k\ge DP(\Pi_{Q^*}-\Pi_{Q^*})Q_k=0$ and the second inequality is due to the hypothesis $Q_k^L- Q^*\le Q_k-Q^*$ and the fact that $A_{Q^*}$ is a nonnegative matrix (all elements are nonnegative). The proof is completed by induction.
\end{proof}

Defining $x_k := Q_k^L  - Q^*$ and $A:= A_{Q^*}$,~\eqref{eq:lower-system} can be concisely represented as the \emph{stochastic linear system}
\begin{align}
x_{k + 1} = A x_k + \alpha w_k,\quad x_0 \in {\mathbb R}^n,\quad \forall k \geq 0,\label{eq:linear-system-form}
\end{align}
where $n:= |{\cal S} \times {\cal A}|$, and $w_k\in {\mathbb R}^n$ is a stochastic noise. The noise $w_k$ has the zero mean, and is bounded. It is formally proved in the following lemma.
\begin{lemma}\label{lemma:bound-W}
We have
\begin{enumerate}
\item ${\mathbb E}[w_k] = 0$;

\item ${\mathbb E}[\left\| {w_k } \right\|_\infty  ] \le \sqrt{W_{\max }}$;

\item ${\mathbb E}[\left\| {w_k } \right\|_2 ] \le \sqrt{W_{\max }}$;

\item ${\mathbb E}[w_k^T w_k ] \le \frac{9}{{(1 - \gamma )^2 }} = :W_{\max }$.
\end{enumerate}
for all $k\geq 0$.
\end{lemma}
\begin{proof}
For the first statement, we take the conditional expectation on~\eqref{eq:w} to have ${\mathbb E}[w_k |x_k ] = 0$. Taking the total expectation again with the law of total expectation leads to the first conclusion. Moreover, the conditional expectation, ${\mathbb E}[w_k^T w_k |Q_k ]$, is bounded as
\begin{align*}
&{\mathbb E}[w_k^T w_k\, |Q_k ]\\
 =& {\mathbb E}[\left\| {w_k } \right\|_2^2\, |Q_k ]\\
=& {\mathbb E}[\left\| {(e_{a_k}  \otimes e_{s_k} )\delta _k  - (DR + \gamma DP\Pi _{Q_k } Q_k  - DQ_k )} \right\|_2^2\, |Q_k ]\\
=& {\mathbb E}[\delta _k^2\, |Q_k ] - \left\| {DR + \gamma DP\Pi _{Q_k } Q_k  - DQ_k } \right\|_2^2\\
\le& {\mathbb E}[\delta _k^2\, |Q_k ]\\
=& {\mathbb E}[r_k^2\, |Q_k ] + {\mathbb E}[2r_k \gamma (e_{s'} )^T \Pi _{Q_k } Q_k\, |Q_k ]\\
&+ {\mathbb E}[ - 2r_k (e_{a_k}  \otimes e_{s_k} )^T Q_k\, |Q_k ]\\
&+ {\mathbb E}[ - 2\gamma (e_{s_k'} )^T \Pi _{Q_k } Q_k (e_{a_k}  \otimes e_{s_k} )^T Q_k\, |Q_k ]\\
&+ {\mathbb E}[\gamma (e_{s_k'} )^T \Pi _{Q_k } Q_k \gamma (e_{s_k'} )^T \Pi _{Q_k } Q_k\, |Q_k ]\\
& + {\mathbb E}[(e_{a_k}  \otimes e_{s_k} )^T Q_k (e_{a_k}  \otimes e_{s_k} )^T Q_k \, |Q_k ]\\
\le& 1 + 2\gamma {\mathbb E}[|r_k |\times |(e_{s_k'} )^T \Pi _{Q_k } Q_k |\,|Q_k ]\\
& + 2{\mathbb E}[|r_k |\times  |(e_{a_k}  \otimes e_{s_k} )^T Q_k |\,|Q_k ]\\
&+ 2\gamma {\mathbb E}[|(e_{s_k'} )^T \Pi _{Q_k } Q_k |\times |(e_{a_k}  \otimes e_{s_k} )^T Q_k |\,|Q_k ]\\
& + \gamma ^2 {\mathbb E}[|(e_{s_k'} )^T \Pi _{Q_k } Q_k |\times |(e_{s_k'} )^T \Pi _{Q_k } Q_k |\,|Q_k ]\\
&+ {\mathbb E}[|(e_{a_k}  \otimes e_{s_k} )^T Q_k |\times |(e_{a_k}  \otimes e_{s_k} )^T Q_k |\,|Q_k ]\\
\le& \frac{9}{{(1 - \gamma )^2 }}=:W_{\max},
\end{align*}
where $\delta_k$ is defined in~\eqref{eq:td-error}, and the last inequality comes from Assumptions~\ref{assumption:bounded-reward}-\ref{assumption:bounded-Q0}, and~\cref{lemma:bounded-Q}. Taking the total expectation, we have the fourth result. Next, taking the square root on both sides of ${\mathbb E}[\left\| {w_k } \right\|_2^2 ] \le W_{\max}$, one gets ${\mathbb E}[\left\| {w_k } \right\|_\infty  ] \le {\mathbb E}[\left\| {w_k } \right\|_2 ] \le \sqrt {{\mathbb E}[\left\| {w_k } \right\|_2^2 ]}  \le \sqrt {W_{\max}}$, where the first inequality comes from $\left\|  \cdot  \right\|_\infty   \le \left\|  \cdot  \right\|_2$. This completes the proof.
\end{proof}

To proceed further, let us define the covariance of the noise
\[
{\mathbb E}[w_k w_k^T ] = :W_k= W_k^T \succeq 0.
\]
An important quantity we use in the main result is the maximum eigenvalue, $\lambda _{\max } (W_k)$, whose bound can be easily established as follows.
\begin{lemma}\label{lemma:bound-W2}
The maximum eigenvalue of $W_k$ is bounded as
\begin{align*}
\lambda _{\max } (W_k) \le W_{\max}
\end{align*}
for all $k \geq 0$, where $W_{\max} > 0$ is given in~\cref{lemma:bound-W}.
\end{lemma}
\begin{proof}
The proof is completed by noting $\lambda _{\max } (W_k) \le {\rm tr}(W_k) = {\rm tr}({\mathbb E}[w_k w_k^T ]) = {\mathbb E}[{\rm tr}(w_k w_k^T )] = {\mathbb E}[w_k^T w_k ] \le W_{\max}$, where the last inequality comes from~\cref{lemma:bound-W}, and the second equality uses the fact that the trace is a linear function. This completes the proof.
\end{proof}

As a next step, we investigate how the autocorrelation matrix, ${\mathbb E}[x_k x_k^T ]$, propagates over the time. In particular, the autocorrelation matrix is updated through the linear recursion
\[
{\mathbb E}[x_{k + 1} x_{k + 1}^T ] = A {\mathbb E}[x_k x_k^T ]A^T  + \alpha ^2 W_k,
\]
where ${\mathbb E}[w_k w_k^T ] = W_k$. Defining $X_k := {\mathbb E}[x_k x_k^T ], k \geq 0$, it is equivalently written as
\begin{align}
X_{k + 1}  = AX_k A^T  + \alpha^2 W_k,\quad \forall k\geq 0,\label{eq:correlation-update}
\end{align}
with $X_0  := x_0x_0^T$. The following lemma proves the fact that the trace of $X_k$ is bounded, which will be used for the main development.
\begin{lemma}[Bounded trace]\label{lemma:basic1}
We have the following bound:
\[
{\rm tr} ( X_k) \le \frac{{9 n^2 \alpha }}{{d_{\min } (1 - \gamma )^3 }} + \| x_0\|_2^2 n^2 \rho ^{2k}.
\]
\end{lemma}
\begin{proof}

We first bound $\lambda _{\max } (X_k )$ as follows:
\begin{align*}
\lambda _{\max } (X_k ) \le& \alpha ^2 \sum_{i = 0}^{k - 1} {\lambda _{\max } (A^i W_{k - i - 1} (A^T )^i )}  + \lambda _{\max } (A^k X_0 (A^T )^k )\\
\le& \alpha ^2 \sup _{j \ge 0} \lambda _{\max } (W_j )\sum_{i = 0}^{k - 1} {\lambda _{\max } (A^i (A^T )^i )}\\
&  + \lambda _{\max } (X_0 )\lambda _{\max } (A^k (A^T )^k )\\
=& \alpha ^2 \sup _{j \ge 0} \lambda _{\max } (W_j ) \sum_{i = 0}^{k - 1} {\| {A^i }\|_2^2 }  + \lambda _{\max } (X_0 )\| {A^k } \|_2^2\\
\le& \alpha ^2  W_{\max } n\sum_{i = 0}^{k - 1} {\| {A^i } \|_\infty ^2 }  + n\lambda _{\max } (X_0 )\| {A^k } \|_\infty ^2\\
\le& \alpha ^2 W_{\max } n\sum_{i = 0}^{k - 1} {\rho ^{2i} }  + n\lambda _{\max } (X_0 )\rho ^{2k}\\
\le& \alpha ^2 W_{\max } n \lim_{k \to \infty } \sum_{i = 0}^{k - 1} {\rho ^{2i} }  + n\lambda _{\max } (X_0 )\rho ^{2k}\\
\le& \frac{{\alpha ^2 W_{\max } n}}{{1 - \rho ^2 }} + n\lambda _{\max } (X_0 )\rho ^{2k}\\
\le& \frac{{\alpha ^2 W_{\max } n}}{{1 - \rho }} + n\lambda _{\max } (X_0 )\rho ^{2k},
\end{align*}
where the first inequality is due to $A^i W_{k-i-1}(A^T )^i \succeq 0$ and $A^k X_0 (A^T )^k \succeq 0$, the third inequality comes from~\cref{lemma:bound-W2} and $\left\|  \cdot  \right\|_2  \le \sqrt n \left\|  \cdot  \right\|_\infty$, the fourth inequality is due to~\cref{lemma:max-norm-system-matrix}, the sixth and last inequalities come from $\rho \in (0,1)$. On the other hand, since $X_k \succeq 0$, the diagonal elements are nonnegative. Therefore, we have ${\rm{tr}}(X_k ) \le n\lambda _{\max } (X_k )$. Combining the last two inequalities lead to
\begin{align*}
{\rm{tr}}(X_k ) \le n\lambda _{\max } (X_k ) \le \frac{{\alpha ^2 W_{\max } n^2 }}{{1 - \rho }} + n^2 \lambda _{\max } (X_0 )\rho ^{2k}.
\end{align*}

Moreover, noting the inequality $\lambda _{\max } (X_0) \le {\rm tr}(X_0 ) = {\rm tr}(x_0 x_0^T ) = \left\| {x_0 } \right\|_2^2$, and plugging $\rho = 1-\alpha d_{\min}(1-\gamma)$ into $\rho$ in the last inequality, one gets the desired conclusion.
\end{proof}

Now, we are ready to present the main results. In the first result, we provide a finite-time bound on the state error of the lower comparison system.
\begin{theorem}\label{thm:main1}
For any $k \geq 0$, we have
\begin{align}
{\mathbb E}[\| Q_k^L  - Q^*\|_2 ] \le \frac{{3\alpha ^{1/2} |{\cal S} \times {\cal A}|}}{{d_{\min }^{1/2} (1 - \gamma )^{3/2} }} + |{\cal S} \times {\cal A}| \|Q_0^L  - Q^*\|_2 \rho ^k.
\label{eq:8}
\end{align}
\end{theorem}
\begin{proof} Noting the relations
\begin{align*}
{\mathbb E}[\|Q_k^L  - Q^*\|_2^2 ] =& {\mathbb E}[(Q_k^L  - Q^* )^T (Q_k^L  - Q^* )]\\
 =& {\mathbb E}[{\rm tr}((Q_k^L  - Q^* )^T (Q_k^L - Q^* ))]\\
 =& {\mathbb E}[{\rm tr}((Q_k^L  - Q^* )(Q_k^L - Q^* )^T )]\\
 =& {\mathbb E}[{\rm tr}(X_k )],
\end{align*}
and using the bound in~\cref{lemma:basic1}, one gets
\begin{align}
{\mathbb E}[\|Q_k^L  - Q^*\|_2^2 ] \le \frac{{9\alpha n^2 }}{{d_{\min } (1 - \gamma )^3 }} + n^2 \|Q_0^L  - Q^*\|_2^2 \rho ^{2k}.\label{eq:9}
\end{align}
Taking the square root on both side of the last inequality, using the subadditivity of the square root function, the Jensen inequality, and the concavity of the square root function, we have the desired conclusion.
\end{proof}

The first term on the right-hand side of~\eqref{eq:8} can be made arbitrarily small by reducing the step-size $\alpha \in (0,1)$. The second bound exponentially vanishes as $k \to \infty$ at the rate of $\rho = 1 - \alpha d_{\min} (1 - \gamma) \in (0,1)$. Therefore, it proves the exponential convergence of the mean-squared error of the lower comparison system up to a constant bias. In the next subsection, we will investigate an analysis of an upper comparison system.

\subsection{Upper comparison system}
Now, let us consider the stochastic linear switching system~\cite{lee2021discrete}
\begin{align}
Q_{k+1}^U-Q^*= A_{Q_k}(Q_k^U-Q^*)+\alpha w_k,\quad Q_0^U-Q^*\in {\mathbb R}^{|{\cal S}\times {\cal A}|},
\label{eq:upper-system}
\end{align}
 where the stochastic noise $w_k$ is kept the same as the original system. We will call it the \emph{upper comparison system}.
\begin{proposition}[{\cite{lee2021discrete}}]\label{thm:upper-bound}
Suppose $Q_0^U-Q^*\ge Q_0-Q^*$, where $\geq $ is used as the element-wise inequality. Then, $Q_k^U-Q^*\ge Q_k-Q^*$ for all $k \ge 0$.
\end{proposition}
\begin{proof}
Suppose the result holds for some $k \ge 0$. Then, we have
\begin{align*}
(Q_{k+1}- Q^*)=& A_{Q_k}(Q_k- Q^*)+b_{Q_k}+\alpha w_{k}\\
\leq& A_{Q_k}(Q_k- Q^*)+\alpha w_{k}\\
\leq& A_{Q_k}(Q_k^U - Q^*)+\alpha w_k\\
=& Q_{k+1}^U-Q^*,
\end{align*}
where we used the fact that $b_{Q_k}=D(\gamma P\Pi_{Q_k} Q^*-\gamma P\Pi_{Q^*} Q^*)\le D(\gamma P\Pi_{Q^*} Q^*-\gamma P\Pi_{Q^*} Q^*)=0$ in the first inequality. The second inequality is due to the hypothesis $Q_k^U-Q^*\ge Q_k-Q^*$ and the fact that $A_{Q_k}$ is a nonnegative matrix. The proof is completed by induction.
\end{proof}

According to~\cref{thm:upper-bound}, the trajectory of the stochastic linear switching system in~\eqref{eq:upper-system} bounds that of the original system~\eqref{eq:swithcing-system-form} from above. Then, with the notation $x_k := Q_k^U  - Q^*$,~\eqref{eq:upper-system} can be concisely represented as the \emph{stochastic switching linear system}
\begin{align}
x_{k + 1} = A_{\sigma_k} x_k + \alpha w_k,\quad x_0 \in {\mathbb R}^n,\quad \forall k \geq 0,\label{eq:switching-system-form}
\end{align}
where $n:= |{\cal S}\times {\cal A}|$, ${\sigma _k} \in \{ 1,2, \ldots ,|\Theta |\}$ is the switching signal at time $k\geq 0$ determined by ${\sigma _k} = \varphi ({\pi _k})$ with ${\pi _k}(\cdot): = \arg {\min _{a \in {\cal A}}}{Q_k}( \cdot ,a) \in \Theta$, $\varphi :\Theta  \to \{ 1,2, \ldots ,|\Theta |\} $ is a one-to-one mapping from a deterministic policy $\pi \in \Theta$ to an integer in $\{ 1,2, \ldots ,|\Theta |\}$, and matrices $A_i, i\in \{ 1,2, \ldots ,|\Theta |\}$ are defined in~\eqref{eq:12}.

Compared to the lower comparison system~\eqref{eq:linear-system-form}, which is linear,~\eqref{eq:switching-system-form} is a switching system, which is much more complicated due to the dependency on $Q_k$ and $Q_k^U$. In particular, the system matrix $A_{Q_k}$ switches according to the change of $Q_k$, which  depends probabilistically on $Q_k^U$. Therefore, if we take the expectation on both sides, it is not possible to separate $A_{Q_k}$ and the state $Q_k^U-Q^*$ unlike the lower comparison system, making it much harder to analyze the stability of the upper comparison system. Therefore, the analysis used for the upper comparison system cannot be directly applied, i.e., the autocorrelation matrix ${\mathbb E}[x_k x_k^T ]$ cannot be obtained by using the simple linear recursion given in~\eqref{eq:correlation-update}. To overcome this difficulty, in the next subsection, we instead study an error system by subtracting the lower comparison system~\cite{lee2021discrete} from the upper comparison system.

\subsection{Analysis of original system}

In the previous subsections, we have introduced upper and lower comparison systems, and provided bounds on the corresponding expected state errors in~\cref{thm:main1}. Since the states of the lower and upper comparison systems bound the state of the original system from below and above, respectively, i.e.,
\begin{align}
Q_k^L  - Q^*  \le Q_k  - Q^*  \le Q_k^U  - Q^*,\label{eq:7}
\end{align}
one can prove that the mean state error of the original system is also bounded in terms of those of the upper and lower comparison systems.

However, as discussed previously, compared to the lower comparison system~\eqref{eq:linear-system-form}, which is linear,~\eqref{eq:switching-system-form} is a switching system, which is much more complicated due to the dependency on $Q_k$ and $Q_k^U$. To circumvent such a difficulty, we instead study an error system by subtracting the lower comparison system~\cite{lee2021discrete} from the upper comparison system:
\begin{align}
Q_{k+1}^U- Q_{k+1}^L
=A_{Q_k}(Q_k^U-Q_k^L) +B_{Q_k}(Q_k^L-Q^*), \label{eq:error-system}
\end{align}
where
\begin{align}
B_{Q_k}:=A_{Q_k}-A_{Q^*}=\alpha \gamma DP(\Pi_{Q_k}-\Pi_{Q^*}).\label{eq:B-matrix}
\end{align}

Here, the stochastic noise $\alpha w_k$ is canceled out in the error system. Matrices $(A_{Q_k}, B_{Q_k})$ switch according to the external signal $Q_k$, and $Q_k^L-Q^*$ can be seen as an external disturbance.
The key insight is as follows: if we can prove the stability of the error system, i.e., $Q_k^U-Q_k^L \to 0$ as $k\to\infty$, then since $Q_k^L \to Q^*$ as $k \to \infty$, we have $Q_k^U \to Q^*$ as well. Keeping this picture in mind, we can establish the following bound on the expected error ${\mathbb E}[\left\| {Q_k  - Q^* } \right\|_\infty  ]$.
\begin{theorem}[Convergence]\label{thm:main2}
For all $k \geq 0$, we have
\begin{align}
{\mathbb E}[\left\| {Q_k  - Q^* } \right\|_\infty  ] \le& \frac{{9 d_{\max } |{\cal S} \times {\cal A}|\alpha ^{1/2} }}{{d_{\min }^{3/2} (1 - \gamma )^{5/2} }} + \frac{{2|{\cal S} \times {\cal A}|^{3/2} }}{{1 - \gamma }}\rho ^k\nonumber\\
& + \frac{{4\alpha \gamma d_{\max } |{\cal S} \times {\cal A}|^{2/3} }}{{1 - \gamma }}k\rho ^{k - 1}\label{eq:3}
\end{align}
\end{theorem}
\begin{proof}
Taking norm on the error system in~\eqref{eq:error-system}, we get
\begin{align*}
\| {Q_{k + 1}^U  - Q_{k + 1}^L } \|_\infty   \le& \| A_{Q_k } \|_\infty  \| Q_k^U  - Q_k^L \|_\infty \\
&+ \| B_{Q_k } \|_\infty  \| Q_k^L  - Q^*\|_\infty \\
 \le& \rho \|Q_k^U  - Q_k^L\|_\infty   + 2\alpha \gamma d_{\max } \|Q_k^L  - Q^*\|_\infty
\end{align*}
where the second inequality is due to~\cref{lemma:max-norm-system-matrix} and the definition in~\eqref{eq:B-matrix}.
Combining the last inequality with that in~\cref{thm:main1} and $\| {Q_k^L  - Q^* } \|_\infty   \le \| {Q_k^L  - Q^* } \|_2$ yields
\begin{align*}
{\mathbb E}[\| Q_{i + 1}^U  - Q_{i + 1}^L \|_\infty  ] \le& \rho {\mathbb E}[ \|Q_i^U  - Q_i^L\|_\infty]\\
& + 2\alpha \gamma d_{\max } \frac{{3|{\cal S} \times {\cal A}|\alpha ^{1/2} }}{{d_{\min }^{1/2} (1 - \gamma )^{3/2} }}\\
& + 2\alpha \gamma d_{\max } \|Q_0^L  - Q^*\|_2 |{\cal S} \times {\cal A}|\rho ^i
\end{align*}
for all $i\geq 0$. Applying the inequality successively for these values from $i=0$ to $i=k$ leads to
\begin{align*}
{\mathbb E}[{\| {Q_k^U - Q_k^L} \|_\infty }] \le& {\rho ^k}{\mathbb E}[{\| {Q_0^U - Q_0^L} \|_\infty }]\\
& + \frac{{6\gamma {d_{\max }}|{\cal S} \times {\cal A}|{\alpha ^{1/2}}}}{{d_{\min }^{3/2}{{(1 - \gamma )}^{5/2}}}}\\
& + k{\rho ^{k - 1}}2\alpha \gamma {d_{\max }}{ \| {Q_0^L - {Q^*}} \|_2}| {\cal S} \times {\cal A}|.
\end{align*}
Next, letting $Q_0^U  = Q_0^L = Q_0$ yields
\begin{align}
{\mathbb E}[\| Q_k^U  - Q_k^L\|_\infty  ] \le& \frac{{6\gamma d_{\max } |{\cal S} \times {\cal A}|\alpha ^{1/2} }}{{d_{\min }^{3/2} (1 - \gamma )^{5/2} }}\nonumber\\
& + k\rho ^{k - 1} 2\alpha \gamma d_{\max } \|Q_0  - Q^*\|_2 |{\cal S} \times {\cal A}|.\label{eq:10}
\end{align}
Using $\|Q_0  - Q^*\|_2  \le |{\cal S} \times {\cal A}|^{1/2} \|Q_0  - Q^*\|_\infty   \le |{\cal S} \times {\cal A}|^{1/2} \frac{2}{{1 - \gamma }}$ further leads to
\begin{align}
{\mathbb E}[\| Q_k^U  - Q_k^L \|_\infty  ] \le& \frac{{6\gamma d_{\max } |{\cal S} \times {\cal A}|\alpha ^{1/2} }}{{d_{\min }^{3/2} (1 - \gamma )^{5/2} }}\nonumber\\
& + k\rho ^{k - 1} 4 \alpha \gamma d_{\max } \frac{{|{\cal S} \times {\cal A}|^{2/3} }}{{1 - \gamma }}\label{eq:2}
\end{align}

On the other hand, one can prove the inequalities
\begin{align}
{\mathbb E}[\left\| {Q_k  - Q^* } \right\|_\infty ] =& {\mathbb E}[ \|Q_k  - Q_k^L  + Q_k^L  - Q^* \|_\infty ]\nonumber\\
 \le& {\mathbb E}[\| Q_k^L  - Q^* \|_\infty ] + {\mathbb E}[ \| Q_k  - Q_k^L \|_\infty ]\nonumber\\
 \le& {\mathbb E}[\| Q_k^L  - Q^* \|_\infty ] + {\mathbb E}[ \| Q_k^U  - Q_k^L \|_\infty ]\label{eq:11}
\end{align}
where the last inequality comes from the fact that $Q_k^U  - Q_k^L  \ge Q_k  - Q_k^L  \ge 0$ holds.
Combining the last inequality with that in~\cref{thm:main1} leads to
\begin{align*}
{\mathbb E}[\| Q_k  - Q^*\|_\infty  ] \le& {\mathbb E}[ \|Q_k^L  - Q^*\|_\infty  ] + {\mathbb E}[\|Q_k^U  - Q_k^L\|_\infty  ]\\
\le& \frac{{3\alpha ^{1/2} |{\cal S} \times {\cal A}|}}{{d_{\min }^{1/2} (1 - \gamma )^{3/2} }}\\
& + \frac{{2|{\cal S} \times {\cal A}|^{3/2} }}{{1 - \gamma }}\rho ^k  + {\mathbb E}[\|Q_k^U  - Q_k^L\|_\infty  ].
\end{align*}
Moreover, combining the above inequality with~\eqref{eq:2} yields the desired conclusion.
\end{proof}

Note that the first term in~\eqref{eq:3} is the constant error due to the constant step-size, which is scaled according to $\alpha \in (0,1)$. The second term in~\eqref{eq:3} is due to the gap between lower comparison system and original system, and the third term in~\eqref{eq:3} is due to the gap between upper comparison system and original system. The second term $O(\rho^k)$ exponentially decays, and the third term $O(k \rho^{k-1})$ also exponentially decays while the speed is slower than the second term due to the additional linearly increasing factor. The upper bound in~\eqref{eq:3} can be converted to looser but more interpretable forms as follows.
\begin{corollary}\label{thm:main3}
For any $k \geq 0$, we have
\begin{align}
{\mathbb E}[\|Q_k  - Q^*\|_\infty  ] \le& \frac{{9 d_{\max } |{\cal S} \times {\cal A}|\alpha ^{1/2} }}{{d_{\min }^{3/2} (1 - \gamma )^{5/2} }} + \frac{{2|{\cal S} \times {\cal A}|^{3/2} }}{{1 - \gamma }}\rho ^k \nonumber \\
& + \frac{{4\alpha \gamma d_{\max } |{\cal S} \times {\cal A}|^{2/3} }}{{1 - \gamma }}\frac{{ - 2}}{{\ln (\rho )}}\rho ^{\frac{{ - 1}}{{\ln (\rho )}} - 1} \rho ^{k/2}\label{eq:4}
\end{align}
and
\begin{align}
{\mathbb E}[\left\| {Q_k  - Q^* } \right\|_\infty  ] \le& \frac{{9 d_{\max } |{\cal S} \times {\cal A}|\alpha ^{1/2} }}{{d_{\min }^{3/2} (1 - \gamma )^{5/2} }}\nonumber\\
& + \frac{{2|{\cal S} \times {\cal A}|^{3/2} }}{{1 - \gamma }}\rho ^k\nonumber\\
& + \frac{{8\gamma d_{\max } |{\cal S} \times {\cal A}|^{2/3} }}{{1 - \gamma }}\frac{1}{{d_{\min } (1 - \gamma )}}\rho ^{k/2 - 1}\label{eq:5}
\end{align}

\end{corollary}
\begin{proof}
In~\eqref{eq:3}, we focus on the term $k\rho ^{k - 1}  = k\rho ^{k/2 + k/2 - 1}  = k\rho ^{k/2 - 1} \rho ^{k/2}$. Let $f(x) = x\rho ^{x/2}  = x\rho ^{x/2}$. Checking the first-order optimality condition
\begin{align*}
\frac{{df(x)}}{{dx}} = \frac{d}{{dx}}x\rho ^{x/2}  = \rho ^{x/2}  + x\frac{1}{2}\rho ^{x/2} \ln (\rho ) = 0,
\end{align*}
it follows that its maximum point is $x = \frac{{ - 2}}{{\ln (\rho )}}$, and the corresponding maximum value is $f\left( {\frac{{ - 2}}{{\ln (\rho )}}} \right) = \frac{{ - 2}}{{\ln (\rho )}}\rho ^{\frac{{ - 1}}{{\ln (\rho )}}}$. Therefore, we have the bounds $k\rho ^{k - 1}  = k\rho ^{k/2} \rho ^{ - 1} \rho ^{k/2}  \le \frac{{ - 2}}{{\ln (\rho )}}\rho ^{\frac{{ - 1}}{{\ln (\rho )}} - 1} \rho ^{k/2}$. Combining this bound with~\eqref{eq:3}, one gets the first bound in~\eqref{eq:4}.
To obtain the second inequality in~\eqref{eq:5}, we use the relation $1 - \frac{1}{x} \le \ln x \le x - 1,\forall x > 0$ to obtain
\begin{align*}
\frac{1}{{\ln (\rho ^{ - 1} )}}\rho ^{\frac{1}{{\ln (\rho ^{ - 1} )}}}  \le& \frac{1}{{1 - \frac{1}{{\rho ^{ - 1} }}}}\rho ^{\rho ^{ - 1}  - 1}\\
\le& \frac{1}{{\alpha d_{\min } (1 - \gamma )}}\rho ^{\frac{1}{{1 - \alpha d_{\min } (1 - \gamma )}} - 1}\\
\le& \frac{1}{{\alpha d_{\min } (1 - \gamma )}},
\end{align*}
where the last inequality uses $\alpha \in (0,1)$ in~\cref{assumption:step-size}.
Combining the above bound with~\eqref{eq:4},~\eqref{eq:5} follows. This completes the proof.
\end{proof}

\begin{remark}
A probabilistic error bound can be derived from the expected error bound by leveraging various concentration inequalities, such as the Markov inequality. For instance, using the Markov inequality, we have
\begin{align*}
{\mathbb P}[{\left\| {{Q_k} - {Q^*}} \right\|_\infty } < \varepsilon ] \ge& 1 - \frac{{9\gamma {d_{\max }}|{\cal S} \times {\cal A}|{\alpha ^{1/2}}}}{{\varepsilon d_{\min }^{3/2}{{(1 - \gamma )}^{5/2}}}}\\
&-\frac{{2|{\cal S} \times {\cal A}|^{3/2}}}{{\varepsilon (1 - \gamma )}}{\rho ^k}\\
& - \frac{{8\gamma {d_{\max }}|{\cal S} \times {\cal A}{|^{2/3}}}}{{\varepsilon (1 - \gamma )}}\frac{1}{{{d_{\min }}(1 - \gamma )}}{\rho ^{k/2 - 1}}.
\end{align*}
The right-hand side converges to one as $k \to \infty$ and $\alpha \to 0$.
\end{remark}

\section{Comparative analysis}

\begin{table}[h!]
\caption{Comparative analysis of several results: $t_{\rm cover}$ is the cover time; $t_{\rm mix}$ is the mixing time; $\tilde {\cal O}$ ignores the polylogarithmic factors}
\begin{center}
{\tiny
\begin{tabular}{c c c}
\hline Method & Sample complexity & Observation model\\
\hline
Ours & $\tilde {\cal O}\left( {\frac{{\gamma ^2 d_{\max }^2 |{\cal S} \times {\cal A}|^2 }}{{\varepsilon ^2 d_{\min }^4 (1 - \gamma )^6 }}} \right)$ & i.i.d. \\

Lee et. al. \cite{lee2021discrete} & $\tilde {\cal O}\left( \frac{d_{\max}^4 |{\cal S}\times {\cal A}|^4}{\varepsilon^4 d_{\min}^6 (1-\gamma)^{10}} \right)$ & i.i.d. \\

Beck et. al. \cite{beck2012error} & $\tilde {\cal O}\left( {\frac{t_{\rm cover}^3 |{\cal S}\times {\cal A}|}{(1-\gamma)^5 \varepsilon^2}} \right)$ & non-i.i.d. \\

Li et. al. \cite{li2020sample} & $\tilde {\cal O}\left( \frac{1}{d_{\min}(1-\gamma)^5 \varepsilon^2} + \frac{t_{\rm mix}}{d_{\min}(1-\gamma)}\right)$ & non-i.i.d. \\

Chen et. al. \cite{chen2021lyapunov} & $\tilde {\cal O}\left(\frac{1}{d_{\min}^3(1-\gamma)^5\varepsilon^2}\right)$ & non-i.i.d.\\

Qu et. al. \cite{qu2020finite} & $\tilde {\cal O}\left( {\frac{{{t_{{\rm{mix}}}}|{\cal S} \times {\cal A}{|^2}}}{{{{(1 - \gamma )}^5}{\varepsilon ^2}}}} \right)$  & non-i.i.d. \\

Even-Dar et. al. \cite{even2003learning} & $\tilde {\cal O}\left( {\frac{{{{({t_{{\rm{cover}}}})}^{\frac{1}{{1 - \gamma }}}}}}{{{{(1 - \gamma )}^4}{\varepsilon ^2}}}} \right)$  & non-i.i.d. \\

\hline
\end{tabular}}
\end{center}\label{table1}
\end{table}
The sample complexities of Q-learning, as analyzed and reported in various existing works~\cite{beck2012error,li2020sample,chen2021lyapunov,qu2020finite,even2003learning,lee2021discrete}, are summarized in~\cref{table1}, where $t_{\rm cover}$ represents the cover time, $t_{\rm mix}$ denotes the mixing time, and $\tilde {\cal O}$ omits the polylogarithmic factors, and the proof of the sample complexity based on the proposed convergence bound in~\cref{thm:main3} is given in Appendix. It is worth noting that most of these analyses adopt non-i.i.d. observation models. To account for these non-i.i.d. observation models, the cover time assumptions are considered in~\cite{beck2012error,even2003learning}, while the mixing time assumptions are employed in~\cite{li2020sample,qu2020finite,chen2021lyapunov}. These sample complexity bounds are derived under various assumptions and conditions, making it generally impractical to make direct comparisons among them.
However, it is worth noting that the proposed sample complexity does not appear to be consistently tighter than existing approaches, which represents a limitation of our method. Regarding the step-size conditions, our proposed finite-time analysis allows for a step-size $\alpha \in (0,1)$, which is more flexible compared to~\cite{chen2021lyapunov,beck2012error,li2020sample}. This is because the constant step-sizes used in~\cite{chen2021lyapunov,beck2012error,li2020sample} impose more restrictive ranges for the finite-time analysis. In our view, the main advantage of the proposed approach lies in its introduction of a unique switching system and control viewpoints. Building upon these perspectives, we have developed clear and simpler analysis frameworks for finite-time error bounds. These perspectives not only offer simplicity but also yield valuable insights into Q-learning. The proposed techniques, based on the foundational principles of systems and control theory, render the overall analysis more accessible and intuitive particularly for people with a background in control theory.

\section{Conclusion}\label{sec:conclusion}
In this paper, we have revisited the switching system framework in~\cite{lee2021discrete} to analyze the finite-time convergence bound of Q-learning. We have improved the analysis in~\cite{lee2021discrete} by replacing the average iterate with the final iterate, which is simpler and more common in the literature. The proposed finite-time error bounds are more general than most existing bounds for the constant step-size Q-learning in terms of the allowable range of step-sizes. Besides, the proposed analysis potentially offers additional insights on analysis of Q-learning, and complements existing approaches. Potential future topics include finite-time analysis of SARSA, double Q-learning, and actor-critic using similar dynamic system viewpoints.

\bibliographystyle{IEEEtran}
\bibliography{reference}

\section{Appendix}
Using the bound in~\eqref{eq:5}, to achieve ${\mathbb E}[\|Q_k- Q^* \|_\infty] < \varepsilon$, a sufficient condition is
\begin{align*}
&\frac{{9\gamma d_{\max } |{\cal S} \times {\cal A}|\alpha ^{1/2} }}{{d_{\min }^{3/2} (1 - \gamma )^{5/2} }} + \frac{{2|{\cal S} \times {\cal A}|^{3/2} }}{{1 - \gamma }}\rho ^k\\
& + \frac{{8\gamma d_{\max } |{\cal S} \times {\cal A}|^{2/3} }}{{1 - \gamma }}\frac{1}{{d_{\min } (1 - \gamma )}}\rho ^{k/2 - 1}  \le \varepsilon.
\end{align*}
The inequality holds if each of the three terms is bounded by $\varepsilon /3$.
For the first term, the bound $\frac{{9\gamma d_{\max } |{\cal S} \times {\cal A}|\alpha ^{1/2} }}{{d_{\min }^{3/2} (1 - \gamma )^{5/2} }} \le \varepsilon /3$ leads to $\alpha  \le \frac{{\varepsilon ^2 d_{\min }^3 (1 - \gamma )^5 }}{{729\gamma ^2 d_{\max }^2 |{\cal S} \times {\cal A}|^2 }}$.
Therefore, we first let $\alpha$ equal to the right-hand side of the above inequality.
For the second term, the bound $\frac{{2|{\cal S} \times {\cal A}|^{3/2} }}{{1 - \gamma }}\rho ^k  \le \varepsilon /3$ leads to $k \ge \frac{{\ln \left( {\frac{{\varepsilon (1 - \gamma )}}{{6|{\cal S} \times {\cal A}|^{3/2} }}} \right)}}{{\ln (\rho )}}$.
Using the relation $1 - \frac{1}{x} \le \ln x \le x - 1,\forall x > 0$, a sufficient condition for the above condition is
\begin{align*}
k \ge \frac{{729\gamma ^2 d_{\max }^2 |{\cal S} \times {\cal A}|^2 }}{{\varepsilon ^2 d_{\min }^4 (1 - \gamma )^6 }}\ln \left( {\frac{{6|{\cal S} \times {\cal A}|^{3/2} }}{{\varepsilon (1 - \gamma )}}} \right).
\end{align*}
For the last term, the bound $\frac{{8\gamma d_{\max } |{\cal S} \times {\cal A}|^{2/3} }}{{d_{\min } (1 - \gamma )^2 }}\rho ^{k/2 - 1}  \le \frac{\varepsilon }{3}$ yields $k \ge \frac{{2\ln \left( {\frac{{\varepsilon d_{\min } (1 - \gamma )^2 }}{{24\gamma d_{\max } |{\cal S} \times {\cal A}|^{2/3} }}} \right)}}{{\ln (\rho )}}$.
Again, using the relation $1 - \frac{1}{x} \le \ln x \le x - 1,\forall x > 0$ results in the sufficient condition $k \ge \ln \left( {\frac{{24\gamma d_{\max } |{\cal S} \times {\cal A}|^{2/3} }}{{\varepsilon d_{\min } (1 - \gamma )^2 }}} \right)\frac{{1458\gamma ^2 d_{\max }^2 |{\cal S} \times {\cal A}|^2 }}{{\varepsilon ^2 d_{\min }^4 (1 - \gamma )^6 }}$. Combining the two bounds leads to the desired conclusion.

\end{document}